\documentclass[a4paper,USenglish,cleveref, autoref]{lipics-v2019}

\newcommand\OPT{\ensuremath{\mathrm{OPT}}\xspace}
\newcommand\TOP{\ensuremath{\mathrm{TOP}}\xspace}

\newcommand\MIN{\ensuremath{\mathrm{MIN}}\xspace}
\newcommand\Tc{\ensuremath{\mathcal{T}}\xspace}
\newcommand\Fc{\ensuremath{\mathcal{F}}\xspace}

\newcommand\Sc{\ensuremath{\mathcal{S}}\xspace}
\newcommand\Xc{\ensuremath{\mathcal{X}}\xspace}

\newcommand{\lbl}[2]{\footnotesize {\color{#1} #2}}
\usepackage{pgf, tikz, tkz-graph}
\usepackage{graphicx}
\usepackage[title]{appendix}
\usepackage{comment}
\usepackage{algorithm}
\usepackage{algpseudocode}
\usepackage{todonotes}
\usepackage{xspace}
\usetikzlibrary{arrows,positioning, calc,trees, automata}

\tikzstyle{vertex}=[draw,fill=black!15,circle,minimum size=20pt,inner sep=0pt]

\title{Approximation algorithms for the vertex-weighted grade-of-service Steiner tree problem} 

\titlerunning{Approximation algorithms for the V-GSST problem}

\author{Faryad Darabi Sahneh}{Department of Computer Science, University of Arizona, Tucson, AZ }{faryad@cs.arizona.edu}{https://orcid.org/0000-0002-1825-0097}{}

\author{Alon Efrat}{Department of Computer Science, University of Arizona, Tucson, AZ }{alon@cs.arizona.edu}{https://orcid.org/0000-0002-1825-0097}{}

\author{Stephen Kobourov}{Department of Computer Science, University of Arizona, Tucson, AZ}{kobourov@cs.arizona.edu}
{https://orcid.org/0000-0002-0477- 2724}{}

\author{Spencer Krieger}{Department of Computer Science, University of Arizona, Tucson, AZ}{skrieger@cs.arizona.edu}{}{}

\author{Richard Spence}{Department of Computer Science, University of Arizona, Tucson, AZ }{rcspence@cs.arizona.edu}{https://orcid.org/0000-0003-4382-466X}{}

\authorrunning{F. Darabi Sahneh et al.}

\Copyright{Faryad Darabi Sahneh, Alon Efrat, Stephen Kobourov, Richard Spence}



\category{}

\relatedversion{https://arxiv.org/abs/1811.11700}

\supplement{}

\acknowledgements{We thank Clara Cromey, Anand Vastrad, and Sridhar Mocherla for their help in discussions related to this problem.}

\nolinenumbers 

\EventEditors{John Q. Open and Joan R. Access}
\EventNoEds{2}
\EventLongTitle{}
\EventShortTitle{}
\EventAcronym{}
\EventYear{}
\EventDate{}
\EventLocation{}
\EventLogo{}
\SeriesVolume{}
\ArticleNo{}
\usepackage{pgf, tikz, tkz-graph}
\usepackage{wrapfig}
\usetikzlibrary{arrows,positioning, calc,trees, automata}

\tikzstyle{vertex}=[draw,fill=black!15,circle,minimum size=20pt,inner sep=0pt]

\begin{document}

\maketitle







\thispagestyle{empty}

\begin{abstract} 

Given a graph $G = (V,E)$ and a subset $T \subseteq V$ of terminals, a \emph{Steiner tree} of $G$ is a tree that spans $T$. In the vertex-weighted Steiner tree (VST) problem, each vertex is assigned a non-negative weight, and the goal is to compute a minimum weight Steiner tree of $G$. Vertex-weighted problems have applications in network design and routing, where there are different costs for installing or maintaining facilities at different vertices.

We study a natural generalization of the VST problem motivated by multi-level graph construction, the \emph{vertex-weighted grade-of-service Steiner tree problem} (V-GSST), which can be stated as follows: given a graph $G$ and terminals $T$, where each terminal $v \in T$ requires a facility of a minimum grade of service $R(v)\in \{1,2,\ldots\ell\}$, compute a Steiner tree $G'$ by installing facilities on a subset of vertices, such that any two vertices requiring a certain grade of service are connected by a path in $G'$ with the minimum grade of service or better. Facilities of higher grade are more costly than facilities of lower grade. Multi-level variants such as this one can be useful in network design problems where vertices may require facilities of varying priority.

While similar problems have been studied in the edge-weighted case, they have not been studied as well in the more general vertex-weighted case. We first describe a simple heuristic for the V-GSST problem whose approximation ratio depends on $\ell$, the number of grades of service. We then generalize the greedy algorithm of [Klein \& Ravi, 1995] to show that the V-GSST problem admits a $(2 \ln |T|)$-approximation, where $T$ is the set of terminals requiring some facility. This result is surprising, as it shows that the (seemingly harder) multi-grade problem can be approximated as well as the VST problem, and that the approximation ratio does not depend on the number of grades of service.

Finally, we show that this problem is a special case of the directed Steiner tree problem and provide an integer linear programming (ILP) formulation for the V-GSST problem.

\end{abstract}

\setcounter{page}{0}

\newpage

\section{Introduction}
Let $G = (V, E)$ be an undirected, connected graph, and let $T \subseteq V$ be a set of terminals. 
A \emph{Steiner tree} is a subtree of~$G$ that spans $T$, possibly including other vertices. In the classical \emph{Steiner tree (ST) problem}, each edge of $G$ has a positive weight, and the goal is to find a Steiner tree of minimum weight. 
The ST problem is NP-hard~\cite{Karp1972}; it is also APX-hard~\cite{Bern1989} and cannot be approximated within a factor of $96/95$ unless P = NP~\cite{Chlebnik2008}.
  
In the vertex- (or node-) weighted Steiner tree (VST) problem, the vertices of the graph are assigned positive weights, rather than the edges. 
The ST problem with edge and/or vertex weights can be formulated as an instance of VST by replacing each edge $uv$ with two edges $\{uw, wv\}$, where the weight of $w$ equals the weight of $uv$.
Vertex-weighted problems have many applications in network routing, where there may be different costs for installing or maintaining facilities at different vertices. The VST problem is provably harder than ST, and cannot be approximated within a factor of $(1-\varepsilon)\ln |T|$ unless P = NP~\cite{feige1998threshold}, via a simple reduction from the set cover problem~\cite{KLEIN1995104}.  
There are nearly best-possible approximation algorithms for the VST problem that achieve an $O(\log|T|)$ approximation ratio~\cite{GUHA199957,KLEIN1995104}.

In many applications, the terminals may require different levels, priorities, or grades of service~\cite{Balakrishnan1994,Charikar2004ToN,Chuzhoy2008,mateus1994algorithm,mirchandani1996MTT}. For example, when connecting cities with a network, larger cities (hubs) often require higher-quality facilities than smaller ones.

\subsection{Problem definition}



We state the problem naturally in terms of facilities and grades of service. Given an undirected graph $G=(V,E)$, let $T \subseteq V$ be a subset of terminals. Each terminal $v\in T$ has a \emph{required} grade of service $R(v) \in \{1,2,\ldots,\ell\}$. The goal is to install facilities on a subset of vertices, where each terminal $v\in T$ contains a facility of grade $R(v)$ or higher, while ensuring connectivity between vertices requiring the same grade of service or higher. 
Given $v \in V$ and $1 \le i \le \ell$, let $c_i(v)$ denote the cost of installing a facility of grade $i$ on vertex $v$. Naturally, facilities of a higher grade of service are more costly, so we stipulate $0 \le c_1(v) \le c_2(v) \le \ldots \le c_{\ell}(v)$ for all $v \in V$.

\begin{definition}[Vertex-weighted Grade-of-Service Steiner Tree (V-GSST) problem] \label{Def: VGSST}
Given an undirected graph $G = (V,E)$, a set of terminals $T \subseteq V$ with required grades of service $R: T \to \{1,2,\ldots,\ell\}$, and installation costs $c_i(v)$, compute a Steiner tree $G'$ of $G$ with (integer) assigned grades of service $y(v)$ such that the following hold:

\begin{itemize}
    \item For all $v \in T$, the assigned grade of service of $v$ is greater than or equal to its required grade of service $R(v)$, i.e., $1 \le R(v) \le y(v) \le \ell$ for all $v \in T$.
    \item For all terminals $u,v \in T$, the $u$-$v$ path in $G'$ uses vertices with assigned grade of service $\min(R(u), R(v))$ or higher. That is, for each $w$ along the $u$-$v$ path in $G'$, we have $y(w) \ge \min(R(u), R(v))$.
\end{itemize}
The cost of a solution is defined as the sum of the costs of all installed facilities in $G'$, namely $\sum_{v \in V(G')} c_{y(v)}(v)$. The \emph{V-GSST problem} is to find a minimum cost subtree $G^*$ of cost $\OPT$.
\end{definition}

We assume w.l.o.g. that edges have zero cost, as an instance with edge and vertex costs can be converted to an instance with only vertex costs. We may define $c_0(v) = 0$ and $y(v) = 0$ for vertices $v \in V \setminus V(G')$; that is, no facility is installed on $v$. Additionally, because edges have zero cost, a solution can be found given only $y(v)|_{v \in V}$ by finding a spanning tree over the vertices $\{v \mid y(v)=\ell\}$, then iteratively contracting the tree and computing a spanning tree over vertices with $y(v) = \ell-1$, and so on. The case $\ell=1$ is the VST problem.

In the VST problem, it is conventional to assume that terminals have zero cost, as the terminals must be involved in any feasible solution~\cite{GUHA199957}. Similarly for the V-GSST problem, we can assume w.l.o.g. that $c_1(v) = \ldots = c_{R(v)}(v) = 0$ for each $v \in T$; this ignores the ``required'' cost $\sum_{v\in T}c_{R(v)}(v)$, which can be helpful in assessing the quality of a solution. In both problems, an instance with positive terminal costs can be converted to an instance with zero terminal costs. For the V-GSST problem, one can create a dummy vertex $v'$ for each $v \in T$ with zero installation cost, add edge $vv'$, and use $v'$ as a terminal instead of $v$. Lastly, we assume w.l.o.g. that there exists $v \in T$ with $R(v) = \ell$; otherwise $\ell$ can be reduced.

\subsection{Related work} \label{section:related}
The (edge-weighted) ST problem admits a simple 2-approximation~\cite{Gilbert1968}, via a minimum spanning tree of the metric closure\footnote{Given $G=(V,E)$, the \emph{metric closure} of $T \subseteq V$ is the complete graph $K_{|T|}$, where edge weights are equal to the lengths of corresponding shortest paths in $G$.} of~${T}$. The linear program (LP)-based approximation algorithm of Byrka et al.~\cite{Byrka2013} gives a ratio of $\rho = \ln 4 +\varepsilon < 1.39$. 
Details about the ST problem and its variants can be found in~\cite{Hauptmann2015, Proemel2002, Winter1987}, with more edge-weighted network design problems in~\cite{goemans1995general}.

Klein and Ravi~\cite{KLEIN1995104} give a greedy $2 \ln |T|$-approximation to the VST problem (see Section~\ref{sec:greedy}) over terminals $T$.
Guha and Khuller~\cite{GUHA199957} improve the approximation ratio to $1.5 \ln |T|$ via minimum-weight ``$3+$ branch spiders;'' however, the algorithm is not practical for large graphs. Demaine et al.~\cite{Demaine} showed that the VST problem admits a polynomial time constant approximation when restricted to planar graphs. 
Other variants of the VST problem have also been studied, including bi-criteria~\cite{moss2007approximation}, multi-commodity~\cite{kortsarz2009approximating}, $k$-connectivity~\cite{nutov2010TCS}, and degree constrained~\cite{ravi2001approximation}. Exact or near-exact approaches for VST based on Lagrangian relaxation have also been studied~\cite{Cordone2006AnEA,engevall1998}.

Chekuri et al.~\cite{chekuri2007approximation} study the similar node-weighted buy-at-bulk network design problem, defined in terms of sending flows $\delta_i$ to node pairs $s_i$-$t_i$. The cost of routing $x_v$ flow through a node $v$ is given by a sub-additive function $f_v()$. The authors show that the single-source problem (NSS-BB) with non-uniform flow costs admits an $O(\log|T|)$-approximation by giving a randomized algorithm for NSS-BB, then derandomizing it using an LP relaxation. Another somewhat related problem is the \emph{online} vertex-weighted Steiner tree problem, in which the terminals $T$ arrive online. At any stage, a subgraph must connect all terminals that have arrived thus far. Naor et al.~\cite{6108170} describe a randomized $O(\log |V| \log^2 |T|)$-approximation algorithm to the online problem. 

Several results are known on multi-level or grade-of-service Steiner tree problems, where edges are weighted. Balakrishnan et al.~\cite{Balakrishnan1994} give a $(4/3)\rho$-approximation algorithm for the 2-level network design problem with proportional edge costs. 
Charikar et al.~\cite{Charikar2004ToN} describe a simple $4\rho$-approximation for the Quality-of-Service (QoS) Multicast Tree problem with proportional edge costs (termed the \emph{rate model}), which is improved to $e \rho$ through randomized doubling. Karpinski et al.~\cite{Karpinski2005} use an iterative contraction scheme to obtain a $2.454\rho$-approximation. Ahmed et al.~\cite{MLST2018} further improve the approximation ratio to $2.351\rho$. Xue et al.~\cite{Xue2001} show that the grade-of-service Steiner tree problem in the Euclidean plane admits $\frac{4}{3}\rho$ and $\frac{5+4\sqrt{2}}{7}\rho$-approximations for 2 and 3 grades, respectively.

Node-weighted problems on graphs can often be converted to equivalent edge-weighted problems, by requiring that the converted graph is directed. 
Segev~\cite{Segev:1987:NST:34241.34242} gives a simple reduction from VST to the directed Steiner tree (DST) problem, where each edge $uv \in E$ is replaced with two directed edges $(u,v)$ and $(v,u)$, and the weight of edge $(u,v)$ equals the weight of its incoming vertex $v$. The DST problem with $k$ terminals admits a $i(i-1)k^{1/i}$-approximation in time $O(n^i k^{2i})$ for fixed $i \ge 1$~\cite{CHARIKAR199973} using a recursive greedy approach, which implies a polynomial time $O(k^{\varepsilon})$-approximation for fixed $\varepsilon > 0$. By setting $i = \log k$, DST can be $O(\log^2 k)$-approximated in quasi-polynomial time. We show that the V-GSST problem is also a special case of DST; see Appendix~\ref{apdx:dst}. 
\vspace{1em}

\subsection{Our contributions} 


We consider simple top-down and bottom-up approaches (Section~\ref{sec-TDBU}) and show that the top-down approach is an $\ell$-approximation to the V-GSST problem, where $\ell$ is the number of grades of service, if one is allowed access to a VST oracle. If one replaces an oracle with an approximation, this gives a polynomial time $O(\ell \log|T|)$-approximation to the V-GSST problem. However, the bottom-up approach can perform arbitrarily badly.

The main result is the following:
\begin{theorem}\label{thm:main}
There is a polynomial time $(2 \ln |T|)$-approximation algorithm for the V-GSST problem with arbitrary costs $c_i(v)$ for each vertex and grade of service.
\end{theorem}
The algorithm, which we refer to as \textsc{GreedyVGSST} (Section~\ref{sec:greedy}), relies on a generalization of the methods by Klein and Ravi~\cite{KLEIN1995104}. This result is surprising, as it shows that the (seemingly harder) multi-level problem can be approximated as well as the VST problem, and the approximation ratio does not depend on the number of grades of service. Unless P = NP, the approximation ratio is within a constant factor of the best possible.

The \textsc{GreedyVGSST} algorithm maintains a set of ``grade-respecting trees'' and carefully merges a subset of these trees so that the newly-formed tree is also grade respecting. The main tool of the analysis is based on the existence of a ``rooted spider decomposition.'' Similar graph decompositions are often used in network design problems~\cite{calinescu2003network,chekuri2007approximation,GUHA199957,KLEIN1995104,kortsarz2009approximating,nutov2010TCS,nutov2010JC,ravi2001approximation}, and we expect them to be applicable to other multi-level network design problems.

Finally, we show that the V-GSST problem is a special case of the directed Steiner tree problem (Appendix~\ref{apdx:dst}), and provide an integer linear programming (ILP) formulation (Appendix~\ref{apdx:ilp}). 


\section{Top-down and bottom-up approaches for V-GSST}
\label{sec-TDBU}
In this section, we describe two simple heuristics, the \emph{top-down} and \emph{bottom-up} approaches.

\subsection{Top-down approach} \label{section:td}
Assume an oracle can compute a minimum-weight VST for an input graph $G$ over terminals $T$ with vertex costs $c(\cdot)$, denoted $E' = VST(G,c,T)$. A top-down approach is as follows: compute a VST of $G$ over terminals $v\in T$ with the highest required grade of service ($R(v) = \ell$), using vertex costs $c_{\ell}(v)$, and install a facility of grade $\ell$ on each vertex spanned by this VST. The cost incurred at this step, which we denote by $\TOP_{\ell}$, equals the total cost of installing these facilities of grade $\ell$. Then, contract this tree into a single terminal with required grade $\ell-1$ and zero cost. Compute a VST of $G$ over the remaining terminals satisfying $R(v) = \ell-1$, and install facilities of grade $\ell-1$ on each vertex spanned by this VST. Continue this process iteratively until a feasible solution is obtained. This approach is summarized in Algorithm~\ref{alg:td}.

\begin{algorithm}
 \caption{Top-down approach}\label{alg:td}
 \begin{algorithmic}[1]
 \Procedure{TopDownVGSST}{$G,c,R$}
 \For{$i={\ell}, \ldots, 1$}
 \State $E_i := VST(G, c_i, \{v : R(v) \ge i\})$ \label{line:td}
 \State Set $y(v) := i$ for each $v$ spanned by $E_i$
 \If{$i > 1$}
 \State Contract subtree $E_{i}$ to a single terminal $t$ with $R(t) = i-1$ and zero cost
 \EndIf
 \EndFor
 \EndProcedure
 \end{algorithmic}
 \end{algorithm}

Let $\TOP$ be the cost of the V-GSST returned by the top-down approach (Algorithm~\ref{alg:td}), let $\TOP_i$ be the total cost of all vertices in the returned solution containing a facility of grade equal to $i$, so that $\TOP = \sum_{i=1}^{\ell} \TOP_i$. 

\begin{proposition} \label{proposition:top-down}
The top-down heuristic (with an oracle) returns a solution whose cost is not worse than $\ell \cdot \OPT$. 
\end{proposition}


The proof is given in Appendix~\ref{apdx:top-down}. The ratio of $\ell$ is tight as shown by example in Figure~\ref{fig:top-down-example-2}. In this example, the top-down approach avoids a non-terminal ``hub'' $v$ whose cost is only slightly more than the cost of all other non-terminals. If an oracle for VST is replaced with an $O(\log |T|)$-approximation, then this gives a polynomial time $O(\ell \log |T|)$-approximation to the V-GSST problem. 

\input{tikz/top-down-example-2.tex}


\subsection{Bottom-up approach}

An analogous ``bottom-up'' approach is to compute a VST over the set of all terminals $T$, using cost function $c_{\ell}(\cdot)$. Installing facilities of grade of service equal to $\ell$ on each vertex of this VST produces a valid solution. We can then demote these vertices' assigned grade of service to locally improve the cost of the solution. However, this poses challenges as $c_{\ell}(\cdot)$ may add vertices $v$ of grade $\ell$ which cannot be demoted, where $c_{\ell}(v) \gg c_{\ell-1}(v)$. 

\section{The \textsc{GreedyVGSST} algorithm}
\label{sec:greedy}
We first review the $(2 \ln |T|$)-approximation algorithm by Klein and Ravi~\cite{KLEIN1995104} for the VST problem (referred to as KR) and then describe the generalization to the V-GSST problem. The KR algorithm maintains a forest $\Fc$; initially, each terminal is a singleton tree. At each iteration, a vertex $v$ as well as a subset $\Sc \subseteq \Fc$ consisting of two or more trees ($|\Sc| \ge 2$) is merged to form a single tree, with the objective of minimizing the \emph{quotient cost} $\frac{c(v)+\sum_{\Tc\in\Sc}d(v,\Tc)} {|\Sc|}$. Here, $d(v,\Tc)$ is the shortest distance from $v$ to any vertex in the tree $\Tc$, excluding endpoint costs. For any given vertex $v$, an optimal subset $\Sc$ and its corresponding quotient cost can be found in polynomial time, as the only subsets $\Sc$ that need to be considered are those consisting of the 2, 3, \ldots, $|\Fc|$ nearest trees from $v$.
The algorithm terminates when $|\Fc|=1$.



\subsection{Setup}
We use the observation that given an instance of V-GSST consisting of a graph $G$, required grades of service $R(v)$, and vertex costs $c_i(v)$,  there exists an optimal solution $G^*$ (in terms of assigned grades of service $y^*(\cdot)$) such that from any vertex $r \in T$ with $R(r) = \ell$, the path from $r$ to any other terminal uses vertices of non-increasing assigned grades of service.

\begin{definition}[Grade-Respecting Tree (GRT)]\label{Def: GRT}
Let $G$ be a graph, and $\Tc$ be a subtree of $G$. Let $y:V({G}) \to \{0,1,\ldots,\ell\}$ be a labeling function that assigns grades of service to vertices in $G$. Let $r \in \Tc$ be a root vertex. We say that $\Tc$ is a \emph{grade-respecting tree} rooted at $r$ if, for all $v \in V(\Tc)$, the path from $r$ to $v$ in $\Tc$ uses vertices of non-increasing grade of service.
\end{definition}

Our generalization to the V-GSST problem relies on maintaining a set $\Fc$ of  GRTs. The notion is similar to that of the rate-of-service Steiner tree~\cite{Xue2001} or QoS Multicast Tree~\cite{Charikar2004ToN} except that we grow and maintain a collection of such trees during the algorithm. The trees in $\Fc$ are not necessarily vertex-disjoint as the same vertex may appear in different GRTs; hence, we avoid the term ``forest.''

For $v \in V(G)$, let $y(v)$ denote the grade of service of $v$ at the current iteration of the algorithm. Initially, $y(v) = R(v)$ for each $v \in T$, and $y(v) = 0$ for $v \in V(G) \setminus T$. The grades $y(\cdot)$ are updated in each iteration, and are returned as output. On each iteration, a subset of the current set $\Fc$ of GRTs is greedily chosen and connected via their roots to form a new GRT $\Tc_{new}$. The set $\Fc$ is updated, as well as the root and grade of service assignments $y(\cdot)$ for vertices in $\Tc_{new}$. 
 The size of $|\Fc|$ is strictly decreasing at each iteration; once $|\Fc|=1$, the resulting GRT is returned as a feasible solution to the V-GSST problem.

To decide which trees to connect, define $\ell$ cost functions $w_1, w_2, \ldots, w_{\ell}: V(G) \to \mathbb{R}_{\ge 0}$ with the interpretation that $w_i(v)$ denotes the ``incremental'' cost of upgrading vertex $v$ from its current grade of service $y(v)$ to $i$. Initially, $w_i(v) = c_i(v)$ for all $i \in \{1,\ldots,\ell\}$ and $v \in V$. 
 The cost $w_i(v)$ is reduced when $v$ is included in $\Tc_{new}$ to reflect that we have already paid a certain cost for $v$.
Additionally, let $d_i(u,v)$ be the cost of a shortest path from $u$ to $v$ using cost function $w_i(\cdot)$, not including the costs of the endpoints $u$ and $v$. As the costs $w_i(\cdot)$ are updated at each iteration, the distances $d_i(u,v)$ also update at each iteration.



\subsection{Initialization}
Initialize $\Fc$ so that each terminal $v \in T$ is its own GRT; initially $|\Fc| = |T|$. Initialize $y(v) = R(v)$ for all $v \in T$, and $y(v) = 0$ for all $v \in V \setminus T$. Lastly, initialize $w_i(v) = c_i(v)$ for all $i \in \{1,\ldots,\ell\}$ and $v \in V$.

\subsection{Iteration step} Each iteration consists of selecting a \emph{root GRT} $\Tc \in\Fc$ rooted at $r$, a \emph{center} $v_c \in V$ (the center may or may not be $r$), an integer $i\leq R(r)$ representing the grade of service that $v_c$ is ``promoted'' to, and a nonempty subset $\Sc = \{\Tc_k\} \subset\Fc$ of GRTs, such that $R(r_k) \le i$ for all roots $r_k$ associated with $\Sc$. By properly connecting $r$ to $v_c$ using facilities of grade $i$, then connecting $v_c$ to each root $r_k$ using facilities of grade $R(r_k)$, we can substitute the $|\Sc|+1$ GRTs with a new GRT $\Tc_{new}$ rooted at $r$. 
Figure~\ref{fig:iteration-example} illustrates an iteration step.

\begin{figure}[th]
\centering
\tikzset{every picture/.style={line width=0.75pt}} 

\begin{tikzpicture}[x=0.75pt,y=0.6pt,yscale=-0.75,xscale=1]

\draw [line width=2]  (176,73) .. controls (176,55.33) and (194.47,41) .. (217.25,41) .. controls (240.03,41) and (258.5,55.33) .. (258.5,73) .. controls (258.5,90.67) and (240.03,105) .. (217.25,105) .. controls (194.47,105) and (176,90.67) .. (176,73) -- cycle ;
\draw    (217.25,73) -- (239.5,62) ;

\draw [shift={(217.25,73)}, rotate = 333.69] [color={rgb, 255:red, 0; green, 0; blue, 0 }  ][fill={rgb, 255:red, 0; green, 0; blue, 0 }  ][line width=0.75]      (0, 0) circle [x radius= 3.35, y radius= 3.35]   ;
\draw    (239.5,62) -- (250.5,73) ;

\draw    (250.5,73) -- (274.5,68) ;

\draw    (274.5,68) -- (303.5,73) ;
\draw [shift={(303.5,73)}, rotate = 9.78] [color={rgb, 255:red, 0; green, 0; blue, 0 }  ][fill={rgb, 255:red, 0; green, 0; blue, 0 }  ][line width=0.75]      (0, 0) circle [x radius= 3.35, y radius= 3.35]   ;

\draw   (348,40.5) .. controls (348,26.97) and (362.89,16) .. (381.25,16) .. controls (399.61,16) and (414.5,26.97) .. (414.5,40.5) .. controls (414.5,54.03) and (399.61,65) .. (381.25,65) .. controls (362.89,65) and (348,54.03) .. (348,40.5) -- cycle ;
\draw   (320,105.5) .. controls (320,91.97) and (334.89,81) .. (353.25,81) .. controls (371.61,81) and (386.5,91.97) .. (386.5,105.5) .. controls (386.5,119.03) and (371.61,130) .. (353.25,130) .. controls (334.89,130) and (320,119.03) .. (320,105.5) -- cycle ;
\draw   (399,91.5) .. controls (399,77.97) and (413.89,67) .. (432.25,67) .. controls (450.61,67) and (465.5,77.97) .. (465.5,91.5) .. controls (465.5,105.03) and (450.61,116) .. (432.25,116) .. controls (413.89,116) and (399,105.03) .. (399,91.5) -- cycle ;
\draw  [dash pattern={on 0.84pt off 2.51pt}] (314,34.2) .. controls (314,19.73) and (325.73,8) .. (340.2,8) -- (446.3,8) .. controls (460.77,8) and (472.5,19.73) .. (472.5,34.2) -- (472.5,112.8) .. controls (472.5,127.27) and (460.77,139) .. (446.3,139) -- (340.2,139) .. controls (325.73,139) and (314,127.27) .. (314,112.8) -- cycle ;
\draw    (303.5,73) -- (324.5,44) ;

\draw    (324.5,44) -- (366.5,50) ;

\draw    (366.5,50) -- (381.25,40.5) ;
\draw [shift={(381.25,40.5)}, rotate = 327.22] [color={rgb, 255:red, 0; green, 0; blue, 0 }  ][fill={rgb, 255:red, 0; green, 0; blue, 0 }  ][line width=0.75]      (0, 0) circle [x radius= 3.35, y radius= 3.35]   ;

\draw    (303.5,73) -- (338.5,76) ;

\draw    (338.5,76) -- (349.5,91) ;
\draw [shift={(349.5,91)}, rotate = 53.75] [color={rgb, 255:red, 0; green, 0; blue, 0 }  ][fill={rgb, 255:red, 0; green, 0; blue, 0 }  ][line width=0.75]      (0, 0) circle [x radius= 3.35, y radius= 3.35]   ;

\draw    (338.5,76) -- (369.5,71) ;

\draw    (369.5,71) -- (392.5,83) ;

\draw    (392.5,83) -- (418.5,81) ;

\draw    (418.5,81) -- (432.25,91.5) ;
\draw [shift={(432.25,91.5)}, rotate = 37.37] [color={rgb, 255:red, 0; green, 0; blue, 0 }  ][fill={rgb, 255:red, 0; green, 0; blue, 0 }  ][line width=0.75]      (0, 0) circle [x radius= 3.35, y radius= 3.35]   ;

\draw (216,83) node   {$r$};
\draw (219,119) node   {$\Tc$};
\draw (305,84) node   {$v_{c}$};
\draw (440,40) node   {$\Sc \subset \Fc$};

\end{tikzpicture}
\caption{Illustration of an iteration step in $\textsc{GreedyVMLST}$ for some set $\Sc$ with $|\Sc|=3$. The root of the newly formed tree $\Tc_{new}$ is $r$.}
\label{fig:iteration-example}
\end{figure}

At each iteration, we select a root GRT, $v_c$, $i$, and $\Sc$ in order to minimize the \emph{cost-to-connectivity} ratio $\gamma$, defined as follows:
\begin{equation} \label{eqn:gamma}
    \gamma=\frac{d_i(r,v_c)+w_i(v_c)+\sum_{k=1}^{|\Sc|}d_{R(r_k)}(v_c,r_k)}{1+|\Sc|}
\end{equation}

The expression $d_i(r,v_c) + w_i(v_c)$ is the cost of upgrading all vertices (including $v_c$) along a shortest $r$-$v_c$ path to grade $i$. The summation $\sum_{k=1}^{|\Sc|} d_{R(r_k)} (v_c, r_k)$ represents the incremental cost of upgrading the vertices from $v_c$ to each root $r_k$ to the appropriate grade of service; this may overcount costs of vertices that appear on multiple center-to-root paths. The denominator $1+|\Sc|$ represents the number of GRTs that are merged.

Once we have selected $r$, $v_c$, $i$, and $\Sc$ that minimizes $\gamma$, we update the assigned grades $y(\cdot)$ for vertices spanned by $\Tc_{new}$ as follows: for each $v$ on a shortest $r$-$v_c$ path, set $y(v) \gets \max(y(v), i)$. For each $v$ on a shortest $v_c$-$r_k$ path, set $y(v) \gets \max(y(v), R(r_k))$. If $v$ appears on multiple center-to-root paths, then $y(v)$ is set to the maximum over all such $R(r_k)$ grades that $v$ connects.

The \textsc{GreedyVGSST} algorithm is summarized in Algorithm~\ref{alg:MLGH}. Figure~\ref{fig:iteration-example-2} shows a concrete example of the execution of \textsc{GreedyVGSST}, using proportional vertex costs to simplify the presentation. Recall we assume $c_1(v) = \ldots = c_{R(v)}(v) = 0$ for each $v \in T$.

 \begin{algorithm}
 \caption{The \textsc{GreedyVGSST} algorithm}\label{alg:MLGH}
 \begin{algorithmic}[1]
 \Procedure{GreedyVGSST}{$G, c, R$}
 \State Initialize $\Fc$ so that each terminal $v \in T$ is a singleton GRT with $y(v)=R(v)$
 \State For each $i = 1, \ldots, \ell$ and $v \in V$, initialize $w_i(v) = c_i(v)$
 \While {$|\Fc|>1$}
 \State Find root GRT $\Tc\in\Fc$, center $v_c$, integer $i$, and $\Sc\subset\Fc$ minimizing $\gamma$.
 \State Update $y(v)$ for $v$ on the $r$-$v_c$ path, as well as for center-to-root paths
 \State Delete $\Tc$ and all GRTs in $\Sc$ from $\Fc$. Add GRT $\Tc_{new}$ to $\Fc$.
 \State Update weight functions: for $v\in V(\Tc_{new})$, set $w_j(v)=0$ if $j\leq y(v),$ and $w_j(v)=w_j(v)-w_{y(v)}(v)$ if $j>y(v)$.
 \EndWhile
 \EndProcedure
 \end{algorithmic}
 \end{algorithm}
 
 \begin{figure}[H]
\centering
\begin{tikzpicture}[scale=0.55]
\def\shift{-5}

\node[align=left] at (-1,1) {(a)};
\node[align=left] at (-1-2*\shift,1) {(b)};
\node[align=left] at (-1,1+\shift) {(c)};
\node[align=left] at (-1-2*\shift,1+\shift) {(d)};


\tikzstyle{every node}=[draw,shape=circle, fill=red!40, inner sep=1.5pt, minimum size=5pt];

\begin{scope}
\node (A) at (0,1) [label=below:\lbl{red}{2}]{5};

\node (B) at (1,1) [fill=white]{7};

\node (C) at (2,1) [label=below:\lbl{red}{1}]{8};

\node (D) at (2.5,3) [fill=white]{4};

\node (E) at (3,1) [fill=white]{3};

\node (F) at (5,3) [label=below:\lbl{red}{2}]{3};

\node (G) at (4,0) [label={below:\lbl{red}{1}}]{6};
        
\node (H) at (4,2)[fill=white]{1};
        
\draw (A) -- (B) -- (C) -- (E) -- (H) -- (F) -- (D) -- (C);
\draw (E) -- (G);
\end{scope}

\begin{scope}[shift={(-2*\shift,0)}]
\tikzstyle{every node}=[draw,shape=circle, fill=red!40, inner sep=1.5pt, minimum size=5pt, line width = 2];

\node (A) at (0,1) [label=below:\lbl{blue}{2}]{5};

\node (C) at (2,1) [label=below:\lbl{blue}{1}]{8};

\node (F) at (5,3) [label=below:\lbl{blue}{2}]{3};

\node (G) at (4,0) [label={below:\lbl{blue}{1}}]{6};
\end{scope}

\begin{scope}[shift={(0,\shift)}]

\node (A) at (0,1) [line width = 2, label=below:\lbl{blue}{2}]{5};

\node (C) at (2,1) [label=below:\lbl{blue}{1}]{8};

\node (E) at (3,1) [fill=white,label=above:$v_c$, label=below:\lbl{blue}{1}]{3};

\node (F) at (5,3) [line width = 2, label=above:$r$,label=below:\lbl{blue}{2}]{3};

\node (G) at (4,0) [label={below:\lbl{blue}{1}}]{6};
        
\node (H) at (4,2)[fill=white,label=below:\lbl{blue}{1}]{1};
        
\draw (F) -- (H) -- (E) -- (C);
\draw (E) -- (G);
\end{scope}

\begin{scope}[shift={(-2*\shift,\shift)}]

\node (A) at (0,1) [line width = 2, label=above:$r$,label=below:\lbl{blue}{2}]{5};

\node (B) at (1,1) [fill=white,label=below:\lbl{blue}{2}]{7};

\node (C) at (2,1) [label=below:\lbl{blue}{2}]{8};

\node (E) at (3,1) [fill=white,label=above:$v_c$, label=below:\lbl{blue}{2}]{3};

\node (F) at (5,3) [label=below:\lbl{blue}{2}]{3};

\node (G) at (4,0) [label={below:\lbl{blue}{1}}]{6};
        
\node (H) at (4,2)[fill=white,label=below:\lbl{blue}{2}]{1};
        
\draw (F) -- (H) -- (E) -- (C);
\draw (A) -- (B) -- (C) -- (E) -- (G);
\end{scope}
\end{tikzpicture}

\caption{Illustration of \textsc{GreedyVGSST}. (a): Input graph with $\ell=2$, with terminals and positive required grades of service $R(\cdot)$ shown in red, and proportional costs (e.g. the terminal $v$ with cost 8 has $(c_1(v), c_2(v)) = (0,8)$, while the non-terminal $u$ with cost 3 has $(c_1(u), c_2(u)) = (3,6)$). (b): Initial set of singleton GRTs $\Fc$ with $|\Fc| = 4$. (c): Result after choosing $r$ and $v_c$ as shown, $i = 1$, and $\Sc$ of size 2, which minimizes $\gamma$. The value of $\gamma$ is $\gamma = \frac{1 \cdot 1 + 1 \cdot 3 + 0 + 0}{1+2} = \frac{4}{3}$. Note that $w_1(v_c), w_2(v_c) = (3,6)$ initially; after the first iteration, we have $(w_1(v_c), w_2(v_c)) = (0,3)$. (d): Choosing $r$ and $v_c$ as shown, $i=2$, and $\Sc$ of size 1 minimizes $\gamma$. In this case, $d_2(r, v_c) = 2 \cdot 7 + 1 \cdot 8 = 22$, $w_2(v_c) = 3$, and $d_2(v_c, r_1) = 1$ where $r_1$ is the root of the single tree in $|\Sc|$. The value of $\gamma$ is $\gamma = \frac{22+3+1}{1 + 1}=13$. Since $|\Fc| = 1$, \textsc{GreedyVGSST} terminates after two iterations with a cost of $2 \cdot 7 + 1 \cdot 8 + 2 \cdot 3 + 2 \cdot 1 = 30$.}
\label{fig:iteration-example-2}
\end{figure}
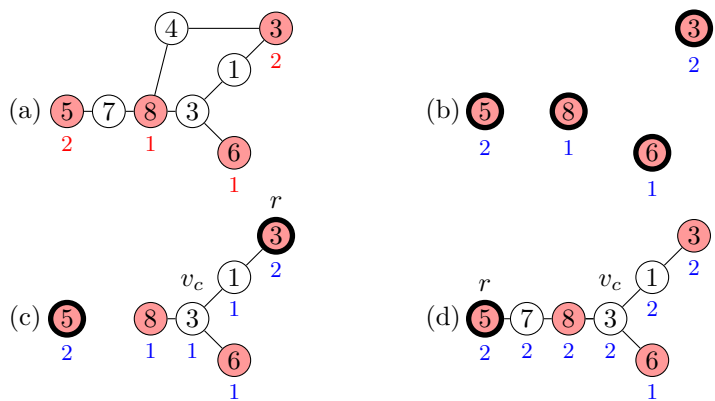
 
Observe that, as edges are not weighted, the actual trees in $\Fc$ at each iteration are not very important until after the final iteration; it suffices to only keep track of roots in $\Fc$ along with the vertices associated with each root.

\begin{lemma} \label{lemma:polynomial} A choice of $r$, $v_c$, $i$, and $\Sc$ that minimizes $\gamma$ can be found in polynomial time.
\end{lemma}
\begin{proof}
For a fixed center $v_c$ and integer $i$, sort all trees in $\Fc$ whose root $r_k$ has $R(r_k) \le i$ by their ``distance'' to $v_c$, namely $d_{R(r_k)}(v_c,r_k)$. The best choice for subset $\Sc$ can be found through checking only subsets with the nearest $2,3,\ldots, |\Fc|$ trees. Therefore, for fixed $v_c$ and $i$, we can find $\Sc$ that minimizes $\frac{w_i(v_c)+\sum_{k=1}^{|\Sc|}d_{R(r_k)}(v_c,r_k)}{|\Sc|}$ in polynomial time. If $i>\max_{k}R(r_k)$, then we can improve $\gamma$ by setting $i=\max_{k} R(r_k)$.

Lastly, find a GRT $\Tc' \in \Fc$ with root $r'$ satisfying $R(r') > i$, whose root is closest to $v_c$ under cost function $w_i(\cdot)$. If such a GRT exists, and if choosing $r'$ as the root lowers $\gamma$, then accept this GRT as the root GRT. Otherwise, use one of the GRTs in $\Sc$ whose root has grade equal to $i$ as the root tree, remove it from $\Sc$, set it as the root GRT $\Tc$, and set $i$ as the maximum grade of the remaining GRT roots in $\Sc$. Therefore, for a fixed center vertex $v_c$ and grade $i$, we can find a root $r$ and subset $\Sc$ that minimizes $\gamma$. As there are $|V|\ell$ choices for $v_c$ and $i$, a choice of $r$, $v_c$, $i$, and $\Sc$ that minimizes $\gamma$ can be found in polynomial time.
\end{proof}

Compared to the KR algorithm, an iteration in the \textsc{GreedyVGSST} algorithm requires determining two new elements: a root GRT $\Tc$, and an integer $i$ representing the grade of service that $v_c$ is upgraded to. However, the above proof indicates that a \textsc{GreedyVGSST} iteration is only $\ell$ times
more expensive than one for the KR algorithm.


 
\subsection{Analysis of {\sc GreedyVGSST}}

The analysis of the KR algorithm~\cite{KLEIN1995104} uses a \emph{spider decomposition} of the optimal VST solution. First, we introduce the notion of a ``locally optimal solution'' with respect to a subset $M$ of vertices. Recall that $y(v)$ denotes the assigned grade of service of $v$.

\begin{definition}[$M-$Optimized GRT]\label{def:m-optimized}
Let $\Tc$ be a GRT rooted at $r$, and let $M \subseteq V(\Tc)$ such that $r\in M$. Then $\Tc$ is \emph{$M$-optimized} if all leaves of $\Tc$ are contained in $M$, and for any vertex $v \in V(\Tc) \setminus M$, we have that $y(v) = \max y(w)$ for all $w \in M$ in the subtree rooted at $v$.
\end{definition}

Thus, an $M$-optimized GRT does not unnecessarily use higher grades of service to reach vertices in $M$ from $r$; see Figure~\ref{fig:lrt}-\ref{fig:m-optimized} for an example. There necessarily exists a minimum V-GSST $G^*$ that is $T$-optimized due to the following simple argument: if there is a vertex $v\notin T,$ for which $y^*(v) > y^*(w)$ for all $w \in T$ in the subtree rooted at $v$, then demoting $y^*(v)$ to equal $\max_w y^*(w)$ over all $w$ in the subtree rooted at $v$ leads to a solution whose cost is less than or equal to that of $G^*$. Further, given a GRT $\Tc$ rooted at $r$, and a set $M \subseteq V(\Tc)$ with $r \in M$, it is not difficult to compute an $M$-optimized subtree of $\Tc$. 


\begin{figure}[h]
\begin{minipage}[b]{0.3\linewidth}
\centering
\tikzset{every picture/.style={line width=0.75pt}} 

\begin{tikzpicture}[x=0.5pt,y=0.5pt,yscale=-1,xscale=1]

\draw    (331.5,14.5) -- (371.5,29.5) ;

\draw    (270.5,28.5) -- (302.5,73.5) ;

\draw    (264.5,82.5) -- (302.5,73.5) ;

\draw    (302.5,73.5) -- (347.5,71.5) ;

\draw    (352.5,126.5) -- (389.5,109.5) ;

\draw    (371.5,29.5) -- (405.5,29.5) ;

\draw    (347.5,71.5) -- (371.5,29.5) ;

\draw    (347.5,71.5) -- (389.5,109.5) ;

\draw    (347.5,71.5) -- (388.5,71.5) ;

\draw  [fill={rgb, 255:red, 0; green, 0; blue, 0 }  ,fill opacity=1 ] (339,71.5) .. controls (339,66.81) and (342.81,63) .. (347.5,63) .. controls (352.19,63) and (356,66.81) .. (356,71.5) .. controls (356,76.19) and (352.19,80) .. (347.5,80) .. controls (342.81,80) and (339,76.19) .. (339,71.5) -- cycle ;
\draw  [fill={rgb, 255:red, 255; green, 255; blue, 255 }  ,fill opacity=1 ] (380,71.5) .. controls (380,66.81) and (383.81,63) .. (388.5,63) .. controls (393.19,63) and (397,66.81) .. (397,71.5) .. controls (397,76.19) and (393.19,80) .. (388.5,80) .. controls (383.81,80) and (380,76.19) .. (380,71.5) -- cycle ;
\draw  [fill={rgb, 255:red, 0; green, 0; blue, 0 }  ,fill opacity=1 ] (381,109.5) .. controls (381,104.81) and (384.81,101) .. (389.5,101) .. controls (394.19,101) and (398,104.81) .. (398,109.5) .. controls (398,114.19) and (394.19,118) .. (389.5,118) .. controls (384.81,118) and (381,114.19) .. (381,109.5) -- cycle ;
\draw  [fill={rgb, 255:red, 255; green, 255; blue, 255 }  ,fill opacity=1 ] (344,126.5) .. controls (344,121.81) and (347.81,118) .. (352.5,118) .. controls (357.19,118) and (361,121.81) .. (361,126.5) .. controls (361,131.19) and (357.19,135) .. (352.5,135) .. controls (347.81,135) and (344,131.19) .. (344,126.5) -- cycle ;
\draw  [fill={rgb, 255:red, 255; green, 255; blue, 255 }  ,fill opacity=1 ] (363,29.5) .. controls (363,24.81) and (366.81,21) .. (371.5,21) .. controls (376.19,21) and (380,24.81) .. (380,29.5) .. controls (380,34.19) and (376.19,38) .. (371.5,38) .. controls (366.81,38) and (363,34.19) .. (363,29.5) -- cycle ;
\draw  [fill={rgb, 255:red, 255; green, 255; blue, 255 }  ,fill opacity=1 ] (294,73.5) .. controls (294,68.81) and (297.81,65) .. (302.5,65) .. controls (307.19,65) and (311,68.81) .. (311,73.5) .. controls (311,78.19) and (307.19,82) .. (302.5,82) .. controls (297.81,82) and (294,78.19) .. (294,73.5) -- cycle ;
\draw  [fill={rgb, 255:red, 255; green, 255; blue, 255 }  ,fill opacity=1 ] (262,28.5) .. controls (262,23.81) and (265.81,20) .. (270.5,20) .. controls (275.19,20) and (279,23.81) .. (279,28.5) .. controls (279,33.19) and (275.19,37) .. (270.5,37) .. controls (265.81,37) and (262,33.19) .. (262,28.5) -- cycle ;
\draw  [fill={rgb, 255:red, 0; green, 0; blue, 0 }  ,fill opacity=1 ] (256,82.5) .. controls (256,77.81) and (259.81,74) .. (264.5,74) .. controls (269.19,74) and (273,77.81) .. (273,82.5) .. controls (273,87.19) and (269.19,91) .. (264.5,91) .. controls (259.81,91) and (256,87.19) .. (256,82.5) -- cycle ;
\draw  [fill={rgb, 255:red, 0; green, 0; blue, 0 }  ,fill opacity=1 ] (397,29.5) .. controls (397,24.81) and (400.81,21) .. (405.5,21) .. controls (410.19,21) and (414,24.81) .. (414,29.5) .. controls (414,34.19) and (410.19,38) .. (405.5,38) .. controls (400.81,38) and (397,34.19) .. (397,29.5) -- cycle ;
\draw  [fill={rgb, 255:red, 0; green, 0; blue, 0 }  ,fill opacity=1 ] (323,14.5) .. controls (323,9.81) and (326.81,6) .. (331.5,6) .. controls (336.19,6) and (340,9.81) .. (340,14.5) .. controls (340,19.19) and (336.19,23) .. (331.5,23) .. controls (326.81,23) and (323,19.19) .. (323,14.5) -- cycle ;
\draw    (235.5,121.5) -- (264.5,82.5) ;

\draw  [fill={rgb, 255:red, 0; green, 0; blue, 0 }  ,fill opacity=1 ] (227,121.5) .. controls (227,116.81) and (230.81,113) .. (235.5,113) .. controls (240.19,113) and (244,116.81) .. (244,121.5) .. controls (244,126.19) and (240.19,130) .. (235.5,130) .. controls (230.81,130) and (227,126.19) .. (227,121.5) -- cycle ;
\draw    (302.5,125.5) -- (264.5,82.5) ;

\draw  [fill={rgb, 255:red, 0; green, 0; blue, 0 }  ,fill opacity=1 ] (294,125.5) .. controls (294,120.81) and (297.81,117) .. (302.5,117) .. controls (307.19,117) and (311,120.81) .. (311,125.5) .. controls (311,130.19) and (307.19,134) .. (302.5,134) .. controls (297.81,134) and (294,130.19) .. (294,125.5) -- cycle ;
\draw    (264.5,82.5) -- (266.5,124.5) ;

\draw  [fill={rgb, 255:red, 255; green, 255; blue, 255 }  ,fill opacity=1 ] (258,124.5) .. controls (258,119.81) and (261.81,116) .. (266.5,116) .. controls (271.19,116) and (275,119.81) .. (275,124.5) .. controls (275,129.19) and (271.19,133) .. (266.5,133) .. controls (261.81,133) and (258,129.19) .. (258,124.5) -- cycle ;

\draw (348,90) node [scale=0.8,color={rgb, 255:red, 0; green, 0; blue, 255 }  ,opacity=1 ] [align=left] {4};
\draw (373,46) node [scale=0.8,color={rgb, 255:red, 0; green, 0; blue, 255 }  ,opacity=1 ] [align=left] {3};
\draw (333,33) node [scale=0.8,color={rgb, 255:red, 0; green, 0; blue, 255 }  ,opacity=1 ] [align=left] {2};
\draw (407,49) node [scale=0.8,color={rgb, 255:red, 0; green, 0; blue, 255 }  ,opacity=1 ] [align=left] {1};
\draw (391,88) node [scale=0.8,color={rgb, 255:red, 0; green, 0; blue, 255 }  ,opacity=1 ] [align=left] {3};
\draw (305,91) node [scale=0.8,color={rgb, 255:red, 0; green, 0; blue, 255 }  ,opacity=1 ] [align=left] {2};
\draw (248,82.5) node [scale=0.8,color={rgb, 255:red, 0; green, 0; blue, 255 }  ,opacity=1 ] [align=left] {2};
\draw (235,140.5) node [scale=0.8,color={rgb, 255:red, 0; green, 0; blue, 255 }  ,opacity=1 ] [align=left] {2};
\draw (268,142.5) node [scale=0.8,color={rgb, 255:red, 0; green, 0; blue, 255 }  ,opacity=1 ] [align=left] {2};
\draw (304,143) node [scale=0.8,color={rgb, 255:red, 0; green, 0; blue, 255 }  ,opacity=1 ] [align=left] {1};
\draw (272,45) node [scale=0.8,color={rgb, 255:red, 0; green, 0; blue, 255 }  ,opacity=1 ] [align=left] {2};
\draw (391,127) node [scale=0.8,color={rgb, 255:red, 0; green, 0; blue, 255 }  ,opacity=1 ] [align=left] {3};
\draw (355,144) node [scale=0.8,color={rgb, 255:red, 0; green, 0; blue, 255 }  ,opacity=1 ] [align=left] {2};
\draw (348,53) node   {$r$};

\end{tikzpicture}
\subcaption{A GRT $\Tc$ with vertices\\ in $M$ shown in black.}\label{fig:lrt}
\end{minipage}
\begin{minipage}[b]{0.3\linewidth}
\centering
\tikzset{every picture/.style={line width=0.75pt}} 

\begin{tikzpicture}[x=0.5pt,y=0.5pt,yscale=-1,xscale=1]

\draw    (331.5,14.5) -- (371.5,29.5) ;

\draw    (264.5,82.5) -- (302.5,73.5) ;

\draw    (302.5,73.5) -- (347.5,71.5) ;

\draw    (371.5,29.5) -- (405.5,29.5) ;

\draw    (347.5,71.5) -- (371.5,29.5) ;

\draw    (347.5,71.5) -- (389.5,109.5) ;

\draw  [fill={rgb, 255:red, 0; green, 0; blue, 0 }  ,fill opacity=1 ] (339,71.5) .. controls (339,66.81) and (342.81,63) .. (347.5,63) .. controls (352.19,63) and (356,66.81) .. (356,71.5) .. controls (356,76.19) and (352.19,80) .. (347.5,80) .. controls (342.81,80) and (339,76.19) .. (339,71.5) -- cycle ;
\draw  [fill={rgb, 255:red, 0; green, 0; blue, 0 }  ,fill opacity=1 ] (381,109.5) .. controls (381,104.81) and (384.81,101) .. (389.5,101) .. controls (394.19,101) and (398,104.81) .. (398,109.5) .. controls (398,114.19) and (394.19,118) .. (389.5,118) .. controls (384.81,118) and (381,114.19) .. (381,109.5) -- cycle ;
\draw  [fill={rgb, 255:red, 255; green, 255; blue, 255 }  ,fill opacity=1 ] (363,29.5) .. controls (363,24.81) and (366.81,21) .. (371.5,21) .. controls (376.19,21) and (380,24.81) .. (380,29.5) .. controls (380,34.19) and (376.19,38) .. (371.5,38) .. controls (366.81,38) and (363,34.19) .. (363,29.5) -- cycle ;
\draw  [fill={rgb, 255:red, 255; green, 255; blue, 255 }  ,fill opacity=1 ] (294,73.5) .. controls (294,68.81) and (297.81,65) .. (302.5,65) .. controls (307.19,65) and (311,68.81) .. (311,73.5) .. controls (311,78.19) and (307.19,82) .. (302.5,82) .. controls (297.81,82) and (294,78.19) .. (294,73.5) -- cycle ;
\draw  [fill={rgb, 255:red, 0; green, 0; blue, 0 }  ,fill opacity=1 ] (256,82.5) .. controls (256,77.81) and (259.81,74) .. (264.5,74) .. controls (269.19,74) and (273,77.81) .. (273,82.5) .. controls (273,87.19) and (269.19,91) .. (264.5,91) .. controls (259.81,91) and (256,87.19) .. (256,82.5) -- cycle ;
\draw  [fill={rgb, 255:red, 0; green, 0; blue, 0 }  ,fill opacity=1 ] (397,29.5) .. controls (397,24.81) and (400.81,21) .. (405.5,21) .. controls (410.19,21) and (414,24.81) .. (414,29.5) .. controls (414,34.19) and (410.19,38) .. (405.5,38) .. controls (400.81,38) and (397,34.19) .. (397,29.5) -- cycle ;
\draw  [fill={rgb, 255:red, 0; green, 0; blue, 0 }  ,fill opacity=1 ] (323,14.5) .. controls (323,9.81) and (326.81,6) .. (331.5,6) .. controls (336.19,6) and (340,9.81) .. (340,14.5) .. controls (340,19.19) and (336.19,23) .. (331.5,23) .. controls (326.81,23) and (323,19.19) .. (323,14.5) -- cycle ;
\draw    (235.5,121.5) -- (264.5,82.5) ;

\draw  [fill={rgb, 255:red, 0; green, 0; blue, 0 }  ,fill opacity=1 ] (227,121.5) .. controls (227,116.81) and (230.81,113) .. (235.5,113) .. controls (240.19,113) and (244,116.81) .. (244,121.5) .. controls (244,126.19) and (240.19,130) .. (235.5,130) .. controls (230.81,130) and (227,126.19) .. (227,121.5) -- cycle ;
\draw    (302.5,125.5) -- (264.5,82.5) ;

\draw  [fill={rgb, 255:red, 0; green, 0; blue, 0 }  ,fill opacity=1 ] (294,125.5) .. controls (294,120.81) and (297.81,117) .. (302.5,117) .. controls (307.19,117) and (311,120.81) .. (311,125.5) .. controls (311,130.19) and (307.19,134) .. (302.5,134) .. controls (297.81,134) and (294,130.19) .. (294,125.5) -- cycle ;

\draw (348,90) node [scale=0.8,color={rgb, 255:red, 0; green, 0; blue, 255 }  ,opacity=1 ] [align=left] {4};
\draw (373,46) node [scale=0.8,color={rgb, 255:red, 0; green, 0; blue, 255 }  ,opacity=1 ] [align=left] {2};
\draw (333,33) node [scale=0.8,color={rgb, 255:red, 0; green, 0; blue, 255 }  ,opacity=1 ] [align=left] {2};
\draw (407,49) node [scale=0.8,color={rgb, 255:red, 0; green, 0; blue, 255 }  ,opacity=1 ] [align=left] {1};
\draw (305,91) node [scale=0.8,color={rgb, 255:red, 0; green, 0; blue, 255 }  ,opacity=1 ] [align=left] {2};
\draw (248,82.5) node [scale=0.8,color={rgb, 255:red, 0; green, 0; blue, 255 }  ,opacity=1 ] [align=left] {2};
\draw (235,140.5) node [scale=0.8,color={rgb, 255:red, 0; green, 0; blue, 255 }  ,opacity=1 ] [align=left] {2};
\draw (304,143) node [scale=0.8,color={rgb, 255:red, 0; green, 0; blue, 255 }  ,opacity=1 ] [align=left] {1};
\draw (391,127) node [scale=0.8,color={rgb, 255:red, 0; green, 0; blue, 255 }  ,opacity=1 ] [align=left] {3};
\draw (348,53) node   {$r$};

\end{tikzpicture}
\subcaption{An $M$-optimized GRT, \\obtained from Fig.~\ref{fig:lrt}.}\label{fig:m-optimized}
\end{minipage}
\begin{minipage}[b]{0.3\linewidth}
\centering
\tikzset{every picture/.style={line width=0.75pt}} 

\begin{tikzpicture}[x=0.5pt,y=0.5pt,yscale=-1,xscale=1]

\draw    (331.5,14.5) -- (371.5,29.5) ;

\draw    (264.5,82.5) -- (302.5,73.5) ;

\draw    (302.5,73.5) -- (347.5,71.5) ;

\draw    (371.5,29.5) -- (405.5,29.5) ;

\draw    (347.5,71.5) -- (371.5,29.5) ;

\draw    (347.5,71.5) -- (389.5,109.5) ;

\draw  [fill={rgb, 255:red, 0; green, 0; blue, 0 }  ,fill opacity=1 ] (339,71.5) .. controls (339,66.81) and (342.81,63) .. (347.5,63) .. controls (352.19,63) and (356,66.81) .. (356,71.5) .. controls (356,76.19) and (352.19,80) .. (347.5,80) .. controls (342.81,80) and (339,76.19) .. (339,71.5) -- cycle ;
\draw  [fill={rgb, 255:red, 0; green, 0; blue, 0 }  ,fill opacity=1 ] (381,109.5) .. controls (381,104.81) and (384.81,101) .. (389.5,101) .. controls (394.19,101) and (398,104.81) .. (398,109.5) .. controls (398,114.19) and (394.19,118) .. (389.5,118) .. controls (384.81,118) and (381,114.19) .. (381,109.5) -- cycle ;
\draw  [fill={rgb, 255:red, 255; green, 255; blue, 255 }  ,fill opacity=1 ] (363,29.5) .. controls (363,24.81) and (366.81,21) .. (371.5,21) .. controls (376.19,21) and (380,24.81) .. (380,29.5) .. controls (380,34.19) and (376.19,38) .. (371.5,38) .. controls (366.81,38) and (363,34.19) .. (363,29.5) -- cycle ;
\draw  [fill={rgb, 255:red, 255; green, 255; blue, 255 }  ,fill opacity=1 ] (294,73.5) .. controls (294,68.81) and (297.81,65) .. (302.5,65) .. controls (307.19,65) and (311,68.81) .. (311,73.5) .. controls (311,78.19) and (307.19,82) .. (302.5,82) .. controls (297.81,82) and (294,78.19) .. (294,73.5) -- cycle ;
\draw  [fill={rgb, 255:red, 0; green, 0; blue, 0 }  ,fill opacity=1 ] (256,82.5) .. controls (256,77.81) and (259.81,74) .. (264.5,74) .. controls (269.19,74) and (273,77.81) .. (273,82.5) .. controls (273,87.19) and (269.19,91) .. (264.5,91) .. controls (259.81,91) and (256,87.19) .. (256,82.5) -- cycle ;
\draw  [fill={rgb, 255:red, 0; green, 0; blue, 0 }  ,fill opacity=1 ] (397,29.5) .. controls (397,24.81) and (400.81,21) .. (405.5,21) .. controls (410.19,21) and (414,24.81) .. (414,29.5) .. controls (414,34.19) and (410.19,38) .. (405.5,38) .. controls (400.81,38) and (397,34.19) .. (397,29.5) -- cycle ;
\draw  [fill={rgb, 255:red, 0; green, 0; blue, 0 }  ,fill opacity=1 ] (323,14.5) .. controls (323,9.81) and (326.81,6) .. (331.5,6) .. controls (336.19,6) and (340,9.81) .. (340,14.5) .. controls (340,19.19) and (336.19,23) .. (331.5,23) .. controls (326.81,23) and (323,19.19) .. (323,14.5) -- cycle ;
\draw    (235.5,121.5) -- (264.5,82.5) ;

\draw  [fill={rgb, 255:red, 0; green, 0; blue, 0 }  ,fill opacity=1 ] (227,121.5) .. controls (227,116.81) and (230.81,113) .. (235.5,113) .. controls (240.19,113) and (244,116.81) .. (244,121.5) .. controls (244,126.19) and (240.19,130) .. (235.5,130) .. controls (230.81,130) and (227,126.19) .. (227,121.5) -- cycle ;
\draw    (302.5,125.5) -- (264.5,82.5) ;

\draw  [fill={rgb, 255:red, 0; green, 0; blue, 0 }  ,fill opacity=1 ] (294,125.5) .. controls (294,120.81) and (297.81,117) .. (302.5,117) .. controls (307.19,117) and (311,120.81) .. (311,125.5) .. controls (311,130.19) and (307.19,134) .. (302.5,134) .. controls (297.81,134) and (294,130.19) .. (294,125.5) -- cycle ;
\draw  [dash pattern={on 0.84pt off 2.51pt}] (256.5,56) .. controls (274.5,57) and (293.5,95) .. (307.5,110) .. controls (321.5,125) and (337.5,125) .. (331.5,138) .. controls (325.5,151) and (228.5,149) .. (218.5,147) .. controls (208.5,145) and (223.5,102) .. (228.5,94) .. controls (233.5,86) and (238.5,55) .. (256.5,56) -- cycle ;
\draw  [dash pattern={on 0.84pt off 2.51pt}] (339.5,55) .. controls (353.5,52) and (354.5,50) .. (378.5,61) .. controls (402.5,72) and (422.5,100) .. (423.5,122) .. controls (424.5,144) and (370.5,141) .. (350.5,111) .. controls (330.5,81) and (325.5,58) .. (339.5,55) -- cycle ;
\draw  [dash pattern={on 0.84pt off 2.51pt}] (308.5,4) .. controls (318.5,2) and (408.5,4) .. (419.5,20) .. controls (430.5,36) and (418.5,59) .. (395.5,56) .. controls (372.5,53) and (331.5,47) .. (316.5,32) .. controls (301.5,17) and (298.5,6) .. (308.5,4) -- cycle ;

\draw (348,90) node [scale=0.8,color={rgb, 255:red, 0; green, 0; blue, 255 }  ,opacity=1 ] [align=left] {4};
\draw (373,46) node [scale=0.8,color={rgb, 255:red, 0; green, 0; blue, 255 }  ,opacity=1 ] [align=left] {2};
\draw (333,33) node [scale=0.8,color={rgb, 255:red, 0; green, 0; blue, 255 }  ,opacity=1 ] [align=left] {2};
\draw (407,49) node [scale=0.8,color={rgb, 255:red, 0; green, 0; blue, 255 }  ,opacity=1 ] [align=left] {1};
\draw (305,91) node [scale=0.8,color={rgb, 255:red, 0; green, 0; blue, 255 }  ,opacity=1 ] [align=left] {2};
\draw (248,82.5) node [scale=0.8,color={rgb, 255:red, 0; green, 0; blue, 255 }  ,opacity=1 ] [align=left] {2};
\draw (235,140.5) node [scale=0.8,color={rgb, 255:red, 0; green, 0; blue, 255 }  ,opacity=1 ] [align=left] {2};
\draw (304,143) node [scale=0.8,color={rgb, 255:red, 0; green, 0; blue, 255 }  ,opacity=1 ] [align=left] {1};
\draw (391,127) node [scale=0.8,color={rgb, 255:red, 0; green, 0; blue, 255 }  ,opacity=1 ] [align=left] {3};

\end{tikzpicture}
\subcaption{Rooted spider decomposition of the GRT from Fig.~\ref{fig:m-optimized}.}\label{fig:decomposition}
\end{minipage}
\caption{A GRT with $|M|=7$ and assigned grades $y(\cdot)$, an $M$-optimized GRT, and a rooted spider decomposition.}
\end{figure}
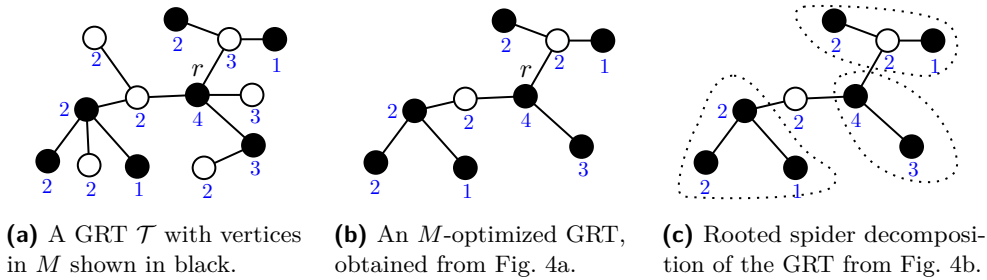

The following definition is due to Klein and Ravi~\cite{KLEIN1995104}:
\begin{definition}[Spider~\cite{KLEIN1995104}]
A \emph{spider} is a tree where at most one vertex has degree greater than two. A spider is identified by its \emph{center}, a vertex from which all paths to the leaves of the spider are vertex-disjoint. A \emph{nontrivial spider} is a spider with at least two leaves.
\end{definition}

A \emph{foot} of a spider is a leaf; if the spider has at least three leaves, then its center is unique, and is also a foot. Klein and Ravi~\cite{KLEIN1995104} show that given a connected graph $G=(V,E)$ and a subset $M\subseteq V$ of vertices, a set of vertex-disjoint nontrivial spiders can be found such that the union of the feet of the spiders contains $M$. We generalize the notion of a spider to the multi-grade setting, which we refer to as a \emph{rooted spider}.

\begin{definition}[Rooted spider]\label{def:rooted_spider}
A \emph{rooted spider} is a GRT $\Tc$ which is also a spider. It is identified by a \emph{center} $v_c$ and a \emph{root} $r$, such that the following properties hold:

\begin{itemize}
    \item The root $r$ is either the center or a leaf of $\Tc$, and the path from $r$ to every vertex in the spider uses vertices of non-increasing grade of service $y(\cdot)$
    \item The paths from $v_c$ to each non-root leaf of $\Tc$ are vertex-disjoint and use non-increasing grades of service
\end{itemize}
\end{definition}

 The resulting tree in Figure~\ref{fig:iteration-example-2} is a rooted spider whose root is distinct from its center.

\begin{lemma} \label{lemma: TailSpider}
Let $\Tc$ be an $M$-optimized GRT rooted at $r$, where $|M| \ge 2$ and $r \in M$. Then $\Tc$ can be decomposed into vertex-disjoint rooted spiders such that the rooted spider leaves and roots belong to $M$ (a rooted spider center may or may not belong to $M$), and the rooted spider leaves, roots, and centers cover the set $M$.
\end{lemma}
\begin{proof}
We use induction on $|M|$. For the base case $|M|=2$, the decomposition consists of a single rooted spider, namely the path in $\Tc$ from $r$ to the other vertex in $M$, where $r$ is also the root of its rooted spider.

For $|M| \ge 3$, find a vertex $v \in \Tc$ with the property that the subtree $\Tc'$ rooted at $v$ contains at least two vertices in $M$, such that $v$ is furthest from $r$ by number of edges. If $v = r$, then $\Tc$ is already a rooted spider, as no other vertex in $V(\Tc)\setminus\{r\}$ can have degree greater than 2; if there existed $w \neq r$ with degree greater than 2, then the subtree rooted at $w$ has at least two leaves (which belong to $M$), contradicting the choice of $v=r$. In this case, set $r$ to be the root and center of its rooted spider.

Thus, we assume $v \neq r$. The vertex $v$ and its subtree forms a rooted spider with center $v$. If $v\in M$, then set $v$ to be the root of its rooted spider. If $v\notin M$, then since $\Tc$ is $M$-optimized, there exists $w \in M$ in the subtree rooted at $v$ with grade of service $y(w) = y(v)$; set $w$ to be the root of its rooted spider. Remove this rooted spider from the original tree $\Tc$, as well as the edge from $v$ to its parent, to produce a smaller tree $\Tc'$.

Let $M' \subset M$ be the set of vertices in $M$ that remain in $\Tc'$ upon removing the subtree rooted at $v$. If $|M'| = 0$, then the subtree rooted at $v$ is the only rooted spider, giving a valid decomposition. If $|M'| = 1$, then $M' = \{r\}$, so connecting $r$ to $v$ produces a single rooted spider with root $r$ and center $v$ as $y(r) \ge y(v)$. Otherwise $|M'| \ge 2$, so we may prune the $r$-$v$ path so that $\Tc'$ is an $M'$-optimized GRT. By the induction hypothesis, $\Tc'$ can be decomposed into vertex-disjoint rooted spiders over $M'$.
\end{proof}

Figure~\ref{fig:decomposition} gives a rooted spider decomposition of the $M$-optimized GRT from Figure~\ref{fig:m-optimized}.

\begin{corollary} \label{corollary:TailSpider}
Let $\Tc$ be an $M$-optimized GRT. Consider a rooted spider decomposition containing $s$ rooted spiders $\Xc_1, \ldots, \Xc_s$, generated using the method in the proof of Lemma~\ref{lemma: TailSpider}.  Let $r_j$ be the root of the $\Xc_j$, and let $L_j = (M \cap V(\Xc_j)) \setminus r_j$ be the vertices in $M$ contained in $\Xc_j$, not including the root $r_j$. Then $\sum_{j=1}^s (1 + |L_j|) = |M|$.
\end{corollary}
\begin{proof}
This statement follows as every vertex in $v \in M$ is either a root of its rooted spider, or is reachable from its spider's center. In particular, the path from the rooted spider's center $v$ to any leaf $w \in M$ does not encounter any other vertices in $M$, as this would contradict the choice of $v$ when computing such a decomposition.
\end{proof}

For the next lemma, we define the following notation. Let $\Fc_n$ ($n\ge 1$) denote the set of GRTs at the beginning of iteration $n$, and let $M_n$ denote the set of GRT roots at the beginning of iteration $n$. Thus, $|\Fc_n| = |M_n|$, and $|\Fc_1| = |T|$. Let $G^*$ be a minimum V-GSST with assigned grades $y^*:V \to \{0,1,\ldots,\ell\}$ and cost $\OPT$.

Consider the $M_n$-optimized GRT obtained by optimizing $G^*$ with respect to $M_n$ and $y^*$. By Lemma~\ref{lemma: TailSpider}, this GRT contains a rooted spider decomposition containing $s$ rooted spiders $\Xc_1, \ldots, \Xc_s$. Consider the $j^{\text{th}}$ rooted spider $\Xc_j$ ($1 \le j \le s$) with root $r_j$, center $v_{c,j}$, and terminals $L_j \subset M_n$ not including its root $r_j$.
Let $c(\Xc_j)$ be the cost of the vertices in $\Xc_j$ within $G^*$, given by $c(\Xc_j) = \sum_{v \in V(\Xc_j)} c_{y^*(v)}(v)$. On the current iteration, a candidate for \textsc{GreedyVGSST} is to select root $r=r_j$, center $v_c=v_{c,j}$, $i = y^*(v_c)$, and $\Sc$ the set of GRTs in $\Fc_n$ whose root is in $L_j$ (so that $|\Sc| = |L_j|$). Let $\widehat{c_j}$ be the cost that \textsc{GreedyVGSST} computes for this candidate (i.e., the numerator in the expression for $\gamma$, eq.~\eqref{eqn:gamma}). We show in Lemma~\ref{Lemma: SpiderCost} that this computed cost is not more than the cost of the $j^{\text{th}}$ rooted spider $\Xc_j$.

\begin{lemma}\label{Lemma: SpiderCost}
Consider the $j^{\text{th}}$ rooted spider $\Xc_j$ in a rooted spider decomposition of the $M_n$-optimized GRT of $G^*$, and consider the candidate choice $r$, $v_c$, $i$, and $\Sc$ as described above with computed cost $\widehat{c_j}$. Then $\widehat{c_j} \le c(\Xc_j)$.
\end{lemma} 
\begin{proof}
Let $p(u,v)$ denote the $u$-$v$ path in the rooted spider $\Xc_j$ not including endpoints $u$ and $v$, and let $c^*(u,v) = \sum_{v \in p(u,v)} c_{y^*(v)}(v)$ denote the sum of vertex costs along path $p(u,v)$. Because the paths from $v_{c,j}$ to the root $r_j$ or to each $w \in L_j$ are vertex-disjoint by Def.~\ref{def:rooted_spider}, we have
\[ c(\Xc_j) \ge c^*(r_j, v_{c,j}) + \sum_{k=1}^{|L_j|} c^*(v_{c,j}, r_k).\]
However, $\widehat{c_j}$ considers the minimum-cost vertex-weighted paths between $r_j$ and $v_{c,j}$, as well as from $v_{c,j}$ to each $r_k \in L_j$. Thus $c(\Xc_j) \ge \widehat{c_j}$ as desired.
\end{proof}

Let $h_n \ge 2$ denote the number of GRTs in $\Fc_n$ that are merged on the $n^{\text{th}}$ iteration, including the root GRT. Let $\Delta C_n$ denote the actual cost incurred on the $n^{\text{th}}$ iteration of \textsc{GreedyVGSST}; for example, $\Delta C_1 = 4$ and $\Delta C_2 = 26$ in the example in Fig.~\ref{fig:iteration-example-2}. Let $\gamma_n$ denote the minimum cost-to-connectivity ratio \emph{computed by the \textsc{GreedyVGSST} algorithm} on the $n^{\text{th}}$ iteration (e.g. $\gamma_1 = \frac{4}{3}$ and $\gamma_2 = 13$ in Fig.~\ref{fig:iteration-example-2}).

\begin{lemma} \label{lemma:iteration-i}
For each iteration $n \ge 1$ of \textsc{GreedyVGSST}, we have $\dfrac{\Delta C_n}{h_n} \le \dfrac{\OPT}{|\Fc_n|}$.
\end{lemma}
\begin{proof}
Fix an iteration $n \ge 1$, and consider the $M_n$-optimized GRT obtained from $G^*$. 
Recall that for $n=1$, $G^*$ is already $T$-optimized ($M_1=T$). 
By Theorem~\ref{lemma: TailSpider}, there exists a rooted spider decomposition over $M_n$, containing $s \ge 1$ rooted spiders $\Xc_1, \ldots, \Xc_s$.


As \textsc{GreedyVGSST} aims to minimize $\gamma$, we necessarily have
\begin{equation}\gamma_n \le \frac{\widehat{c_j}}{1+|L_j|} \underbrace{\le}_{\text{Lemma }\ref{Lemma: SpiderCost}} \frac{c(\Xc_j)}{1+|L_j|}. \label{eqn:gamma_n}
\end{equation}

The computed cost in $\gamma_n$ is greater than or equal to $\Delta C_n$, as the computed cost may overcount vertex costs appearing on multiple center-to-root paths. Hence~\eqref{eqn:gamma_n} implies $\frac{\Delta C_n}{h_n} \le \gamma_n \le \frac{c(\Xc_j)}{1+|L_j|}$ for all rooted spiders $\Xc_j$.


We use the simple algebraic fact that for non-negative numbers $a$, $x_1, \ldots, x_s$, $y_1, \ldots, y_s$, if $a \le \frac{x_i}{y_i}$ for all $1 \le i \le s$, then $a \le (\sum_{j=1}^s x_j)/(\sum_{j=1}^s y_j)$; this is easily verified by writing $ay_i \le x_i$, then summing over $i$. Applying this fact over all rooted spiders, we have $\frac{\Delta C_n}{h_n} \le \frac{\sum_{j=1}^s c(\Xc_j)}{\sum_{j=1}^s (1 + |L_j|)}.$

Observe that $\sum_{j=1}^s c(\Xc_j) \le \OPT$, as the vertices in a rooted spider decomposition of $G^*$ are a subset (not necessarily a proper subset) of $V(G^*)$. The denominator, $\sum_{j=1}^s (1+|L_j|)$, equals $|M_n|=|\Fc_n|$ by Corollary~\ref{corollary:TailSpider}. 
The lemma follows.
\end{proof}

We are ready to prove Theorem~\ref{thm:main}, that \textsc{GreedyVGSST} is a $(2\ln |T|)$-approximation to the V-GSST problem.

\begin{proof}[Proof of Theorem~\ref{thm:main}.]






 Lemma~\ref{lemma:iteration-i} rearranges to $h_n \ge \frac{\Delta C_n}{\OPT}|\Fc_n|$. Division by zero can occur if $\OPT = 0$, e.g. there is a solution where every vertex $v$ is assigned its minimum grade of service $R(v)$. In this case, \textsc{GreedyVGSST} necessarily returns a solution with zero cost, as $\gamma_n = 0$ on every iteration. Hence, we assume $\OPT > 0$. 
 
 Suppose there are $I$ iterations, where $|F_I| \ge 2$ and $|F_{I+1}|:=1$. As $h_n \ge 2$, we equivalently have $\frac{1}{2}h_n \le h_n - 1$. We have $|\Fc_{n+1}| \le |\Fc_n| - (h_n - 1)$ for each iteration $1 \le n \le I$, so Lemma~\ref{lemma:iteration-i} implies the following.

\begin{align}
|\Fc_{n+1}| \le |\Fc_n| - (h_n - 1) &\le |\Fc_n| - \frac{1}{2}h_n \nonumber \\
&\le |\Fc_n|\left(1 - \frac{1}{2} \cdot \frac{\Delta C_n}{\OPT}\right) \label{eq:fn3}
\end{align}

The remainder of the proof 
relies on unraveling the inequalities and taking the logarithm of both sides.

\begin{align}
|\Fc_{I+1}| &\le |\Fc_1| \prod_{n=1}^{I} \left(1 - \frac{1}{2} \cdot \frac{\Delta C_n}{\OPT}\right) \nonumber \\
\ln |\Fc_{I+1}| &\le \ln |\Fc_1| + \sum_{n=1}^{I} \ln \left(1 - \frac{1}{2} \cdot \frac{\Delta C_n}{\OPT}\right) \le \ln |\Fc_1| - \sum_{n=1}^{I} \frac{\Delta C_n}{2\cdot \OPT} \label{eq:ln}\\
\intertext{
where~\eqref{eq:ln} uses the fact that $\ln (1-x) \le -x$ for $x \in [0,1)$. Note that $0 \le \frac{\Delta C_n}{2\cdot \OPT} < 1$ as~\eqref{eq:fn3} implies $1 - \frac{1}{2} \cdot \frac{\Delta C_n}{\OPT} \ge \frac{|F_{n+1}|}{|F_n|} > 0$. Then~\eqref{eq:ln} implies}
\sum_{n=1}^{I} \Delta C_n &\le 2 \cdot \OPT (\ln |\Fc_1| - \ln |\Fc_{I+1}|) = 2 \cdot \OPT (\ln |T| - \ln 1)=2\ln |T| \OPT. \nonumber
\end{align}
completing the proof, as $\sum_{n=1}^I \Delta C_n$ is the cost of the solution {\sc GreedyVGSST} returns.
\end{proof}

It is worth noting that the proof of Theorem~\ref{thm:main} differs from that of Klein and Ravi~\cite{KLEIN1995104} in how a rooted spider decomposition in $G^*$ is considered at each iteration. Once the inequality relating $h_n$ with $\Delta C_n$, $\OPT$, and $|\Fc_n|$ (Lemma~\ref{lemma:iteration-i}) is established, the remainder of the proof is similar to that of Klein and Ravi. Note that $|\Fc_i|$ is strictly decreasing on each iteration, so the number of iterations $I$ is at most $|T|=O(|V|)$. Each iteration can be carried out in polynomial time (Lemma~\ref{lemma:polynomial}), thus \textsc{GreedyVGSST} runs in polynomial time.

In summary, several key techniques allow for the generalization of the KR algorithm~\cite{KLEIN1995104} to the V-GSST problem. First, merging multiple GRTs is non-trivial, as we must ensure that the resulting tree is also a GRT. Our approach is to connect a root GRT, to a center, to a subset $\Sc$ of GRTs. This leads to a time complexity increase by a factor of $\ell$. Second, the analysis of the KR algorithm~\cite{KLEIN1995104} relies on the existence of a spider decomposition in the optimal VST solution. For our analysis, we introduce grade-respecting ``rooted spiders,'' characterized by a root, center, and terminal leaves. Third, instead of contracting subtrees computed on iteration $i$ into ``supernodes,'' in the \textsc{GreedyVGSST} algorithm, it is sufficient to only consider distances between GRT roots, and compute the actual tree at the end.


\section{Conclusions and future work}

We presented a generalization of the VST problem to multiple levels or grades of service and showed that the resulting V-GSST problem admits a $(2\ln |T|)$-approximation, which is surprising as the approximation ratio is optimal (to within a constant), and nearly matches that of the VST problem. The analysis relies on what we call a rooted spider decomposition, which we believe can be of use in other multi-level network design problems.
It will be interesting to investigate whether similar generalizations of other graph sparsification problems can be approximated equally as well as their corresponding single-grade problems, and to evaluate the performance of these algorithms on large real-world graphs.






\newpage
\bibliographystyle{abbrv} 
\bibliography{references,mlst2}

\newpage
\appendix
\section{V-GSST to DST reduction} \label{apdx:dst}
The \emph{directed} Steiner tree problem (DST) is often defined as follows: given a directed graph $G = (V,E)$ with edge weights, a set $T$ of terminals, and a root vertex $r \in V$, compute a Steiner arborescence (directed tree) rooted at $r$, so that there exists an $r$-$v$ path for each $v \in T$. Segev~\cite{Segev:1987:NST:34241.34242} shows that the VST problem can be transformed into an instance of DST, by replacing each edge $uv$ with two directed edges $(u,v)$, $(v,u)$, where the weight of an edge $(u,v)$ equals the weight of its incoming vertex. Set the root $r$ to be any terminal.

Similarly, we show that the V-GSST problem with arbitrary costs $c_i(v)$ can be formulated as an instance of DST. Given a graph $G = (V,E)$, terminals $T$ with required grades of service $R:T \to \{1,2,\ldots,\ell\}$, and vertex costs $c_i(v)$, construct $\ell$ directed copies of $G$, denoted $G^{\ell}$, $G^{\ell-1}$, \ldots, $G^1$. Given $v \in V$, let $v^{i}$ denote the copy of vertex $v$ in $G^i$. For all $\{u,v\} \in E$, set the costs of directed edges $(u^i, v^i)$ and $(v^i, u^i)$ in $G^i$ to be $c_i(v)$ and $c_i(u)$, respectively. Finally, for all $v \in V$ and $i = 1, 2, \ldots, \ell-1$, add the directed edge $(v^{i+1}, v^i)$ with cost zero. The interpretation is that if $v^i$ is spanned in a Steiner arborescence, then $v^{i-1}, \ldots, v^1$ are also spanned.

Finally, assign all vertices $v^{R(v)}$ to be terminals in the transformed instance of DST. For the root vertex, select any vertex $w$ with $R(w) = \ell$, and set $w^{\ell}$ to be the root. Thus, the transformed instance of DST contains $|V|\ell$ vertices, $2|E|\ell + |V|(\ell-1)$ directed edges, and $|T|$ terminals. It is not too hard to show that, given an optimal solution to the V-GSST problem, one can construct an equivalent solution to the DST solution with the same cost, and vice versa.

\section{Proof of Proposition~\ref{proposition:top-down}} \label{apdx:top-down}
\begin{proof}
Here, we can show that $\TOP_i \le \OPT$ for all $1\le i \le \ell$. Note that, when determining a set of vertices to install grade $i$ facilities on (line~\ref{line:td} in Algorithm~\ref{alg:td}), a candidate solution is to consider the set of all vertices containing a facility of grade $i$ or greater within the optimal V-GSST $G^*$, and install facilities of grade $i$ for each vertex in this set. As some vertices may already have facilities of a higher grade installed, the cost incurred when determining $E_i$ is not more than the cost of the facilities of grade $i$ or greater in $G^*$, which is upper bounded by $\OPT$. Hence $\TOP_i \le \OPT$. This immediately implies $\TOP \le \ell \cdot \OPT$.
\end{proof}

\section{Integer linear programming formulation for V-GSST} \label{apdx:ilp}
The following ILP formulation generalizes the cut-based ILP formulation for VST given by Demaine et al.~\cite{Demaine}.

Given some integer $i \le \ell$ and $S \subseteq V$, let $f_i(S) = 1$ if at least one terminal in $T$, but not every terminal, with required grade of service at least $i$, is in $S$. Let $f_i(S) = 0$ otherwise. Let $\Gamma(S)$ be the neighborhood of $S$, defined as the set of vertices adjacent to at least one vertex in $S$ but not in $S$. For $v \in V$ and $i \in \{1,\ldots,\ell\}$, let $x_v^i$ be a binary indicator variable defined as follows:

\[ x_v^i = \begin{cases}
1 & \text{vertex $v$ contains a facility of grade $i$ or higher} \\
0 & \text{otherwise}
\end{cases}
\]
An ILP formulation for the V-GSST problem is as follows:
\begin{align}
    \min &\sum_{v \in V} \sum_{i=1}^{\ell} (w_i(v) - w_{i-1}(v))x_v^i \text{ subject to} \label{eqn:ilp1}\\
    \sum_{v \in \Gamma(S)} x_v^i &\ge f_i(S) & \forall v \in V; S \subseteq V  \label{eqn:ilp2} \\
    x_v^{i} &\ge x_v^{i+1} & \forall v \in V; i \in \{1,2,\ldots,\ell-1\} \label{eqn:ilp3} \\
    x_v^i &\in \{0,1\} & \forall v \in V; i \in \{1,2,\ldots,\ell\} \label{eqn:ilp4}
\end{align}

The objective~\eqref{eqn:ilp1} splits the cost of installing a facility on $v$ into incremental costs, where $w_0(v) = 0$ by definition. Constraint~\eqref{eqn:ilp2} enforces that, for any subset $S$ containing at least one but not all vertices of required grade at least $i$, the cut is crossed by at least one edge in the solution. Constraint~\eqref{eqn:ilp3} enforces that if a facility is installed on $v$ with grade $i+1$ or higher, then it is installed with grade $i$, $i-1$, \ldots or higher. Finally, constraint~\eqref{eqn:ilp4} ensures the $x_v^i$'s are binary.

Such an ILP will only output the $x_v^i$'s, though one can easily extract the assigned grades of service $y(v)$ and construct a valid V-GSST solution given this information.

\end{document}